\documentclass{article}

\usepackage{csquotes}
\usepackage{comment}
\usepackage{arxiv}
\usepackage{amssymb}
\usepackage{amsmath}
\usepackage{amsthm}
\usepackage{graphicx}
\usepackage[numbers]{natbib}

\usepackage{bookmark}

\newtheorem{lemma}{Lemma}
\newtheorem{corollary}{Corollary}
\newtheorem*{question}{Question}
\newtheorem*{conjecture*}{Conjecture}

\newtheorem{theorem}{Theorem}
\newtheorem{proposition}{Proposition}

\newtheorem{definition}{Definition}

\title{Defining binary phylogenetic trees using parsimony: new bounds}

\author{{\hspace{1mm}Mirko Wilde}\\
	Institute of Mathematics and Computer Science\\
	University of Greifswald\\
	Walther-Rathenau-Str. 47, Greifswald, Mecklenburg-Vorpommerania, Germany \\
	\texttt{mirko.wilde@uni-greifswald.de} \\
	\And
	{\hspace{1mm}Mareike Fischer\textsuperscript{*}} \\
	Institute of Mathematics and Computer Science\\
	University of Greifswald\\
	Walther-Rathenau-Str. 47, Greifswald, Mecklenburg-Vorpommerania, Germany \\
	\texttt{email@mareikefischer.de} \\
 \bigskip
 *Corresponding author}

\renewcommand{\shorttitle}{Defining binary phylogenetic trees using parsimony}

\hypersetup{
pdftitle={Defining binary phylogenetic trees using parsimony: new bounds},
pdfauthor={Mirko Wilde, Mareike Fischer},
pdfkeywords={maximum parsimony, phylogenetic tree, Buneman theorem}
}

\begin{document}
\maketitle

\begin{abstract}{Phylogenetic trees are frequently used to model evolution. Such trees are typically reconstructed from data like DNA, RNA, or protein alignments using methods based on criteria like maximum parsimony (amongst others). Maximum parsimony has been assumed to work well for data with only few state changes. Recently, some progress has been made to formally prove this assertion. For instance, it has been shown that each binary phylogenetic tree $T$ with $n \geq 20k$ leaves is uniquely defined by the set $A_k(T)$, which consists of all characters with parsimony score $k$ on $T$. In the present manuscript, we show that the statement indeed holds for all $n \geq 4k$, thus drastically lowering the lower bound for $n$ from $20k$ to $4k$.
However, it has been known that for $n \leq 2k$ and $k \geq 3$, it is not generally true that $A_k(T)$ defines $T$. We improve this result by showing that the latter statement can be extended from $n \leq 2k$ to $n \leq 2k+2$. 
So we drastically reduce the gap of values of $n$ for which it is unknown if trees $T$ on $n$ taxa are defined by $A_k(T)$ from the previous interval of $[2k+1,20k-1]$ to the interval $[2k+3,4k-1]$. Moreover, we close this gap completely for the nearest neighbor interchange (NNI) neighborhood of $T$ in the following sense: We show that as long as $n\geq 2k+3$, no tree that is one NNI move away from $T$ (and thus very similar to $T$) shares the same $A_k$-alignment.}
\end{abstract}



\section{Introduction}\label{sec1}

Reconstructing evolutionary relationships between different species and ultimately even the so-called \enquote{Tree of Life} \cite{ToL}, i.e., the tree describing the relationships of all living species on earth, is one of the big goals in biology. In order to pursue this goal, mathematical tree reconstruction methods are required. Such methods usually take data in the form of aligned DNA, RNA, or protein sequences and then use an optimization criterion to return the \enquote{best} tree, i.e., the tree describing the data best according to a given criterion. One such criterion is maximum parsimony (MP): methods based on this criterion seek to find the tree with the minimal number of nucleotide substitutions (cf. \cite{Fitch1971, Semple2003}). 

However, just as other methods, MP may err, i.e., it may return the wrong tree (e.g., in the so-called \enquote{Felsenstein zone}, cf. \cite{Felsenstein1978,Felsenstein2004}), or it may be indecisive between several trees. On the other hand, it has been observed that in many cases, MP seems to work well when the number of nucleotide substitutions in the data is low \cite{Sourdis1988}. Analyzing and proving this \enquote{folklore knowledge} has inspired various mathematical manuscripts in the recent literature. In particular, the special case of the so-called $A_k(T)$ alignment plays an important role. For a given phylogenetic tree $T$, $A_k(T)$ is the set containing all binary characters that require precisely $k$ nucleotide substitutions on $T$. It is a consequence of the classic Buneman theorem \cite{Buneman1971} in mathematical phylogenetics, which is also known as \enquote{splits equivalence theorem} \cite{Semple2003}, that $A_1(T)$ defines $T$ and that $T$ can be uniquely reconstructed from $A_1(T)$ using the MP criterion. In \cite{Fischer2019}, it was shown that $T$ can be uniquely defined by $A_2(T)$ and moreover, if $n\geq 9$, uniquely reconstructed from $A_2(T)$ using MP. However, in the same manuscript it was also shown that unfortunately, $A_k(T)$ does \textit{not} generally define $T$ for $k\geq 3$ if the number of leaves $n$ of $T$ equals $2k$. On the other hand, more recently it has been shown that $T$ is indeed uniquely defined by $A_k(T)$ whenever $n \geq 20k$ \cite{Fischer2022}. 

So for all cases with $k\geq 3$, the literature so far leaves a big gap between the cases of $2k$ and $20k$: We know that in case that $n=2k$, $T$ is not necessarily defined by $A_k(T)$, but we also know that if $n \geq 20k$, it definitely is. In the present manuscript, we drastically reduce this gap: We show that any binary phylogenetic tree $T$ is uniquely defined by $A_k(T)$ whenever $n \geq 4k$. We obtain this result by exploiting both the classic version of Menger's theorem known from graph theory (cf. \cite{Menger1927,Diestel2017}) as well as a stronger and more recent version of it \cite{Boehme2001}. Thus, our result nicely links modern phylogenetics to classic graph theory, and it can be considered an important first step to proving the conjecture that whenever $k <\frac{n}{4}$, $T$ is the unique maximum parsimony tree of $A_k(T)$, which was stated by  \cite{Fischer2019}  and inspired by \cite{Goloboff2018}. Furthermore, we show that for all $k\geq 3$, there are cases of different trees with $n=2k+1$ as well as $n=2k+2$ and identical $A_k$-alignments, hence also improving the lower bound of the gap, i.e., the range of values of $n$ for which we do \textit{not} know if $A_k(T)$ defines the binary phylogenetic tree $T$ with $n$ leaves. Previously, this gap contained all values of $n$ in the interval $[2k+1,20k-1]$, and our manuscript reduces this interval to $[2k+3,4k-1]$.

Moreover, we even manage to close the gap completely within the so-called NNI neighborhood. In particular, we show that for all $n \geq 2k+3$, we have $A_k(T)\neq A_k(\widetilde{T})$ whenever $\widetilde{T}$ is an NNI neighbor of $T$.
 
In summary, our manuscript drastically reduces the gap of values of $n$ for which we do not know if binary phylogenetic trees on $n$ leaves are uniquely defined by their respective $A_k$-alignments, and we are even able to show that in the NNI neighborhood of any such tree, the gap can be closed affirmatively, in the following sense: The $A_k$-alignment of a given tree $T$ with $n\geq 2k+3$ leaves is indeed unique within its NNI neighborhood.

\section{Preliminaries}

\subsection{Definitions and notation}
Before we can state our results, we need to introduce some basic phylogenetic and graph theoretical concepts. We begin with trees and tree operations.

\subsubsection*{Phylogenetic trees}
In the following, a \textit{phylogenetic $X$-tree} $T=(V,E)$ is a connected acyclic graph without vertices of degree $2$ whose leaves (i.e., vertices whose degree is at most 1) are bijectively labeled by $X$. We may assume without loss of generality that $X=\{1,\ldots,n\}$ for $n \in \mathbb{N}_{\geq 1}$. We call $T$ \textit{binary} if every vertex has degree either $1$ or $3$. Similarly, a \textit{rooted phylogenetic $X$-tree} $T=(V,E)$ is  a connected acyclic graph whose leaves are bijectively labeled by $X$ and with exactly one designated \textit{root} vertex $\rho \in \mathring{T}$. Note that for technical reasons, in the following we will consider a tree consisting of only one vertex to be rooted, too -- in this case, the only vertex is at the same time considered to be the root and the only leaf of the tree.

Two phylogenetic $X$-trees $T=(V,E)$ and $\widetilde{T}=(\widetilde{V},\widetilde{E})$ are \textit{isomorphic} if there exists a bijection $\phi: V \rightarrow \widetilde{V}$ with $\{v,\widetilde{v}\}\in E(T) \Leftrightarrow \{\phi(v), \phi(\widetilde{v})\}\in E(\widetilde{T})$ and $\phi(x) = x$ for all $x\in X$. In other words, $\phi$ is a graph isomorphism preserving the leaf labeling. As is common in mathematical phylogenetics, whenever $T$ and $\widetilde{T}$ are isomorphic, we denote this by $T\cong\widetilde{T}$.

 \subsubsection*{Basic graph theoretical concepts}
 We need some concepts from classical graph theory. Let $G=(V,E)$ be a simple graph and let $A$ and $B$ be subsets of $V$. Then, a path $P$ connecting some vertex in $A$ with some vertex in $B$ and with no interior vertex in either $A$ or $B$ is called an \textit{$A$-$B$-path}. A subset $F$ of $E$ is called an \textit{$A$-$B$-cut set}, or simply \textit{cut set} for short whenever there is no ambiguity, if the graph $G'=(V,E \setminus F)$ resulting from $G$ when edge set $F$ gets deleted contains no $A$-$B$-path. If $\mathcal{P}$ is a collection of $A$-$B$-paths, we denote by $\mathcal{P}(A)$ the union of all sets $V(P)\cap A$  with $P\in \mathcal{P}$ and by $\mathcal{P}(B)$ the union of all sets $V(P)\cap B$ with $P\in \mathcal{P}$, i.e., $\mathcal{P}(A)$ contains all \textit{endpoints} of $P$ that lie in $A$, and $\mathcal{P}(B)$ contains all \textit{endpoints} of $P$ that lie in $B$. In the special case of a phylogenetic $X$-tree with $A=B=X$, we call an $A$-$B$-path with at least one interior vertex a \textit{leaf-to-leaf-path}.
 
 Similar to an $A$-$B$-cut set, an \textit{$A$-$B$-separator} can be defined as follows. For a subset $V^\prime \subseteq V$ we consider the graph $G^\prime=(V\setminus V^\prime,E^\prime)$ induced by $V\setminus V^\prime$, where $E^\prime \subseteq E$ contains all edges both of whose endpoints are contained in $V\setminus V^\prime$. Then, $V^\prime$ is an \textit{$A$-$B$-separator} of $G$ (or separator for short, if there is no ambiguity) for $A,B \subseteq V$, if $G^\prime$ contains no $A$-$B$-path.

\subsubsection*{Phylogenetic tree operations}
When considering various binary phylogenetic $X$-trees, it is often useful to measure their distance using a tree metric. One of the most frequently used such metrics is \textit{$d_{NNI}$}, which simply counts the minimum number of so-called \textit{nearest neighbor interchange (NNI) moves} needed to get from the first tree under consideration to the second one. An NNI move simply takes an inner edge $e$ of a binary phylogenetic $X$-tree $T$ and swaps two of the four subtrees of $T$ which we get when deleting the precisely four edges adjacent to $e$ in a way that the resulting tree is not isomorphic to $T$, i.e., in a way that changes the tree. A tree resulting from $T$ by performing one NNI move is called an \textit{NNI neighbor} of $T$, and all NNI neighbors of $T$ together with $T$ form the \textit{NNI neighborhood} of $T$. Note that this implies that we consider $T$ to belong to its own neighborhood.

Another operation often used to change a binary phylogenetic $X$-tree $T$ is a \textit{cherry reduction}. A cherry $[x,y]$ for $x,y \in X$, $x\neq y$, is a pair of leaves of $T$ adjacent to the same vertex $u$. We distinguish between a \textit{cherry reduction of type 1}, which deletes only one of the two leaves $x$ and $y$ of the cherry, without loss of generality $x$, and suppresses the resulting degree-2 vertex, and a \textit{cherry reduction of type 2}, which deletes both leaves $x$ and $y$ of the cherry as well as their unique neighbor, and subsequently suppresses the resulting degree-2 vertex. Note that other than NNI moves, which transform a binary phylogenetic $X$-tree into another binary phylogenetic $X$-tree, a cherry reduction of type 1 results in a phylogenetic $X\setminus\{x\}$-tree and a cherry reduction of type 2 results in a phylogenetic $X\setminus\{x,y\}$-tree, i.e., both cherry reductions reduce the number of taxa under consideration. However, as long as $\lvert X\rvert \geq 4$, the resulting trees will be binary, too. Such reductions are thus often used in mathematical phylogenetics in inductive proofs. Note that every binary phylogenetic tree with at least three leaves has at least one cherry (\cite[Proposition 1.2.5]{Semple2003}), so that in such trees, both types of cherry reductions can be performed.

\subsubsection*{Characters, $X$-splits and alignments}
In evolutionary biology, we often want to reconstruct phylogenetic trees from a data set on the species in question. In this regard, we often consider \textit{binary characters} $f: X \rightarrow \{a,b\}$. An \textit{extension} of such a character $f$ on a given phylogenetic $X$-tree $T=(V,E)$ is a function $g:V\rightarrow \{a,b\}$ with $g(x) = f(x)$ for $x\in X$. We call an edge $\{v,w\}$ with $g(v)\not= g(w)$ a \textit{changing edge} of $g$. By $ch(g,T)$ we denote the number of changing edges or the \textit{changing number} of $g$ on $T$.\\
Another basic notion in phylogenetics is the following: If $X=A\cup B$ is a bipartition of $X$ into two non-empty subsets, we call $\sigma=A \vert B$ an \textit{$X$-split}. The \textit{size} of an $X$-split $\sigma$ is defined by $\lvert \sigma\rvert= \min\{\lvert A\rvert,\lvert B\rvert\}$. Every $X$-split of size 1 is called \textit{trivial}. 

Note that there is a natural relationship between phylogenetic trees and $X$-splits in the sense that every phylogenetic tree $T$ on taxon set $X$ induces a collection of $X$-splits: If $e$ is an edge of $T$, removing $e$ leads to a forest consisting of two rooted phylogenetic trees $T_A$ and $T_B$ on taxon sets $A$ and $B$, respectively (cf. Fig. \ref{fig:treedecomp}). Clearly, $A\vert B$ is an $X$-split, and we call it an $X$-split \textit{induced by $T$}. We denote the set of all $X$-splits induced by $T$ by $\Sigma(T)$. It is well-known that $\vert\Sigma(T)\vert = 2n-3$ if $T$ is binary and $\vert X\vert\geq 2$ \cite[Proposition 2.1.3]{Semple2003}. Moreover, as clearly all trivial $X$-splits are contained in \textit{every} phylogenetic $X$-tree, we denote by $\Sigma^\ast(T)$ the set of all non-trivial $X$-splits induced by $T$.

Note that not only is there a connection between phylogenetic $X$-trees and $X$-splits, but there is also a connection between $X$-splits and binary characters: Given a binary character $f: X\rightarrow \{a,b\}$, let $A_f=f^{-1}(\{a\})$ and $B_f=f^{-1}(\{b\})$. Then, $A_f\vert B_f$ is an $X$-split, and we call it the $X$-split \textit{induced by $f$}. 

Given some phylogenetic tree $T$ with taxa $X$ and some binary character $f:X\rightarrow \{a,b\}$, we call $l(f,T) := \min_{g} ch(g,T)$ the \textit{parsimony score} of $f$ on $T$. Here, the minimum runs over all extensions $g$ of $f$ on $T$. Note that an extension $g$ of $f$ on $T$ with $ch(g,T)=l(f,T)$ is called \textit{most parsimonious}.

We can extend the definition of the parsimony score to \textit{alignments}, i.e., to a multiset of characters $A=\{f_1, \ldots, f_k\}: l(A, T) = \sum\limits_{i=1}^k l(f_i, T)$. 

Now, the most important concept in our manuscript is a particular alignment: For a given binary phylogenetic tree $T$, we define the alignment $A_k(T)$ as the set of all characters with parsimony score $k$ on $T$. If $A$ is an alignment, then we call it an \textit{$A_k$-alignment} if there exists some binary phylogenetic tree $T$ with $A = A_k(T)$.

\subsection{Known results}
Before we can state our own results, we need to introduce various known results both from the phylogenetic and graph theoretical literature, which we later on use to derive our own findings.

We start by considering edge-disjoint and vertex-disjoint paths and first recall the following useful lemma, which was recently proven in \cite{Fischer2022}. 
\begin{lemma}[Lemma 1 in \cite{Fischer2022}]  \label{lem:Fischer}
Let $T$ be a binary phylogenetic $X$-tree with $\vert X\vert  = n$. Then $T$ has at least $\left\lfloor\frac{n}{2} \right\rfloor$ edge-disjoint leaf-to-leaf-paths.
\end{lemma}

Next, we state Menger's classic theorem, which is well-known from graph theory (\cite{Menger1927}, \cite[Theorem 3.3.1]{Diestel2017}).

\begin{theorem}[Menger's theorem]\label{thm:menger3}
Let $G$ be a graph with vertex set $V$ and $A,B\subset V$. Then the minimum number of vertices needed to separate $A$ from $B$ is equal to the maximum number of vertex-disjoint $A$-$B$-paths.
\end{theorem}

The following classic result from mathematical phylogenetics is based on Menger's theorem.

\begin{proposition}[Corollary 5.1.8 in \cite{Semple2003}]\label{prop:menger0}
Let $T$ be a binary phylogenetic $X$-tree, and let $f: X\rightarrow \{a,b\}$ be a binary character. Then $l(f,T)$ is equal to the maximum number of edge-disjoint $A_f$-$B_f$-paths in $T$.
\end{proposition}

From Proposition \ref{prop:menger0}, one can easily derive the following corollary, which is based on the fact that in \textit{binary} trees, the notions of \enquote{edge-disjoint} and \enquote{vertex-disjoint} coincide when considering leaf-to-leaf paths. 

\begin{corollary}\label{cor:edge=vertex}
    Let $T$ be a binary phylogenetic $X$-tree, and let $f: X\rightarrow \{a,b\}$ be a binary character. Then $l(f,T)$ is equal to the maximum number of vertex-disjoint $A_f$-$B_f$-paths in $T$.
\end{corollary}

\begin{proof} Note that for graphs with maximum degree $3$, and thus in particular for binary phylogenetic trees, we have that two paths $P_1,P_2$ are edge-disjoint only if there is no vertex which is internal vertex in $P_1$ as well as in $P_2$. Thus, the statement follows directly from Proposition \ref{prop:menger0} (using the fact that the endpoints of $A_f$-$B_f$-paths are leaves).
\end{proof}

The argument used in the proof of Corollary \ref{cor:edge=vertex}  implies that -- as we are only considering leaf-to-leaf paths in binary phylogenetic trees in this manuscript unless stated otherwise -- whenever one of the terms \enquote{edge-disjoint} or \enquote{vertex-disjoint} is used, it can replaced with the other one. In particular, whenever a version of Menger's theorem is used to obtain some statement on vertex-disjoint leaf-to-leaf-paths, one automatically obtains an analogous statement on edge-disjoint leaf-to-leaf-paths and vice versa.\\

The following statement, which we need to prove our own results, is a stronger version of Menger's classic theorem. 

\begin{proposition}[adapted from \cite{Boehme2001,Diestel2017}] \label{prop:menger2} 
Let $G$ be a graph with vertex set $V$. Consider some subsets $A$ and $B$ of $V$ such that $A$ cannot be separated from $B$ by a set of fewer than $j$ vertices.  Let $\mathcal{P}$ be a set of $i$ vertex-disjoint $A$-$B$-paths in $G$ with $i<j$ with endpoints $\mathcal{P}(A)$ in $A$ and $\mathcal{P}(B)$ in $B$. Then there exists a set $\mathcal{Q}$ of $j$ vertex-disjoint $A$-$B$-paths in $G$ with $\mathcal{P}(A) \subset \mathcal{Q}(A)$ and $\mathcal{P}(B) \subset \mathcal{Q}(B)$, where $\mathcal{Q}(A)$ denotes the endpoints of $\mathcal{Q}$ in $A$ and $\mathcal{Q}(B)$ denotes the endpoints of $\mathcal{Q}$ in $B$.
\end{proposition}

Note that actually, the version of Proposition \ref{prop:menger2} proven in \cite{Boehme2001} as well as in \cite[Chapter 3.3, $2^{nd}$ proof of Menger's theorem]{Diestel2017} is slightly weaker than the version stated here. In particular, the authors used the case of $i$ vertex-disjoint $A$-$B$-paths to derive the case of $i+1$ vertex-disjoint $A$-$B$-paths. However, the above stated stronger version can be easily derived by an iteration of the arguments used to derive the case $i+1$. 

\par \vspace{0.3cm}

We are now in the position to consider alignment $A_k(T)$, which is one of the most important concepts of the present manuscript. Here we recall the following result, which we seek to generalize in the present manuscript.

\begin{proposition}[adapted from Corollary 1 \cite{Fischer2022} and Proposition 1 in \cite{Fischer2019}]\label{prop:A1A2} Let $k \in \{1,2\}$. Let $T$ and $\widetilde{T}$ be two binary phylogenetic $X$-trees. Then, $T \cong \widetilde{T}$ if and only if $A_k(T)=A_k(\widetilde{T})$.
\end{proposition}

Note that the case of $k=1$ is a direct consequence of the well-known Buneman theorem, which states that $T\cong\widetilde{T}$ if and only if $\Sigma(T)=\Sigma(\widetilde{T})$  \cite{Buneman1971} (see also \cite[Theorem 3.1.4]{Semple2003}), whereas the case $k=2$ can be derived from the first case \cite{Fischer2019}.

We are now finally in a position to turn our attention to new results.

\section{Results}
The aim of this section is threefold: First, we state our main result, which is a generalization of Proposition \ref{prop:A1A2} to all $3\leq k \leq \frac{n}{4}$, or, equivalently, to all $n\geq 4k$ for $k \geq 3$.  

Our second result, however, has a somewhat different flavor: In this setting, we do not analyze the entire space of binary phylogenetic $X$-trees to see if for a given such tree $T$, there is another tree $\widetilde{T}$ with $A_k(T)=A_k(\widetilde{T})$. Instead, we investigate the NNI neighborhood of $T$ and present a lower bound on $n$ such that we can guarantee  $A_k(T)\neq A_k(\widetilde{T}) $ if $\widetilde{T}$ lies in this neighborhood. While this second result is doubtlessly of relevance in its own right, we can also use it to prove our third main result, namely that if $n= 2k+1$ or $n=2k+2$ and $k\geq 3$, there are pairs of binary  phylogenetic $X$-trees, where $X=\{1,\ldots,n\}$, which share the same $A_k$-alignments. Note that it has already been known that there are such cases with $k=\frac{n}{2}$, i.e., $n=2k$, \cite{Fischer2019}, so our new result shows that this problem can still occur for smaller values of $k$ or larger values of $n$. 

So together, our results reduce the \enquote{gap} in the literature quite significantly: Before, it was only known that if $n \leq 2k$, $A_k(T)$ does not necessarily define $T$ \cite{Fischer2019}, and that if $n \geq 20k$, $A_k(T)$ indeed does define $T$ \cite{Fischer2022}, leaving the cases $2k+1\leq n \leq 20k-1$ open. Now, our results narrow this gap down to  $2k+3\leq n \leq 4k-1$ and close it completely for pairs $T$ and $\widetilde{T}$, where  $\widetilde{T}$ is from the NNI neighborhood of a given tree $T$.

We begin with our main result.

\subsection{An extension of Proposition \ref{prop:A1A2}}

It is the main aim of the present manuscript to extend Proposition \ref{prop:A1A2} to other values of $k$. However, it has already been known for some time that the statement does not generally hold for $k\geq 3$. For instance, if the number $n$ of leaves equals $2k$, it is known that there are pairs of trees which share the same $A_k$-alignment \cite{Fischer2019}. 

However, it was recently shown in \cite[Theorem 3]{Fischer2022} that if $n \geq 20k$, $A_k(T)$ defines $T$ in the sense that in this case $T\cong\widetilde{T}$ if and only if $A_k(T)=A_k(\widetilde{T})$ (where $T$ and $\widetilde{T}$ are binary phylogenetic $X$-trees with $\vert X \vert = n \geq 20 k$). 
In the present section, we will drastically improve this lower bound from $n\geq 20k$ to $n\geq 4k$ as stated by the following theorem. 

\begin{theorem} \label{thm:characterization}
Let $k\in \mathbb{N}_{\geq 1}$ and let $n\in \mathbb{N}_{\geq 4k}$. Let $T$ and $\widetilde{T}$ be two binary phylogenetic $X$-trees with $\vert X\vert =n$. Then, $T\cong\widetilde{T}$ if and only if $A_k(T) = A_k(\widetilde{T})$.
\end{theorem}

Before we can prove this  theorem, we first need to state and prove the following useful lemma, which is a simple extension of Lemma \ref{lem:Fischer}.

\begin{lemma}\label{lem:leafchoice}
Let $T$ be a binary phylogenetic tree on taxon set $X$ with $\vert X \vert = n\geq 2$. Let $x\in X$. Then, $T$ contains $p=\left\lfloor\frac{n}{2}\right\rfloor$ edge-disjoint leaf-to-leaf paths $P_1,\ldots,P_{p}$ such that $x$ is an endpoint of $P_1$. 
\end{lemma}

\begin{proof}  We first use Lemma \ref{lem:Fischer} to derive a set of $p$ edge-disjoint leaf-to-leaf paths $P_1^\prime,\ldots,P_{p}^\prime$, where $p =\left\lfloor \frac{n}{2}\right\rfloor\geq 1$, and then show that we can modify these paths so that leaf $x$ is an endpoint of one of the paths. 

We define sets $A$ and $B$ with $A,B\subset X$ such that $A$ contains precisely one endpoint of each path $P_1^\prime,\ldots,P_{p}^\prime$ and $B$ contains the other  endpoint of each path $P_1^\prime,\ldots,P_{p}^\prime$. If $x$ is an endpoint of one of these paths, there is nothing to show. So we now consider the case that $x$ is not an endpoint of these paths and add $x$ to $A$. If then there are leaves left that are not contained in any $P_i^\prime$, we add them to $B$. This way, we ensure that $A\cap B=\emptyset$ and $A\cup B = X$, so $A \vert B$ is an $X$-split.

Now, consider some arbitrary $\{x\}$-$B$-path $P$ in $T$, which is also an $A$-$B$-path as $\{x\}\subseteq A$. Let $j$ be the minimum number of vertices needed to separate $A$ from $B$. Applying Proposition \ref{prop:menger2} to $T$ with $\mathcal{P}:= \{P\}$ (and thus $i=1$), we get a collection $\mathcal{Q}$ of $j$ vertex-disjoint $A$-$B$-paths. Note that in particular, as $\mathcal{P}=\{P\}$ and as $P$ is an $A$-$B$-path with endpoint $x\in A$, we have $x\in \mathcal{P}(A)$. 

Now, by Theorem \ref{thm:menger3} we know that $j$ is also equal to the maximum number of vertex-disjoint $A$-$B$-paths. Moreover, by Corollary \ref{cor:edge=vertex}, $j$ is also equal to the maximum number of edge-disjoint $A$-$B$-paths. By choice of $A \vert B$, we know that there are at least $p$ edge-disjoint $A$-$B$-paths (namely $P_1^\prime,\ldots P_p^\prime$), so $j\geq p$. Therefore, $\mathcal{Q}$ contains at least $p$ edge-disjoint paths $P_1,\ldots,P_p$ with endpoints $\mathcal{P}(A)\subseteq \mathcal{Q}(A)$ by Proposition \ref{prop:menger2}. 

 So $x$ is endpoint of one of the $p$ edge-disjoint paths $P_1,\ldots,P_p$ in $\mathcal{Q}$, and we may assume without loss of generality that $x$ is an endpoint of $P_1$. This completes the proof.
\end{proof}

We are now in the position to prove Theorem \ref{thm:characterization}.

\begin{proof}[Proof of Theorem \ref{thm:characterization}]
As the case $k=1$ follows directly from Proposition \ref{prop:A1A2}, we may in the following assume $k\geq 2$.\footnote{Note that by Proposition \ref{prop:A1A2}, we could even assume $k \geq 3$, but our proof does not depend on this requirement. This also implies that our proof supersedes Proposition 1 of \cite{Fischer2019}.}

Let $n\geq 4k$ and let $T$ and $\widetilde{T}$ be two binary phylogenetic $X$-trees as specified in the theorem. If $T\cong\widetilde{T}$, we clearly have $l(f,T) = l(f,\widetilde{T})$ for all binary characters $f$ on $X$, and thus we also have $A_k(T)=A_k(\widetilde{T})$, which completes the first direction of the proof.

So now we assume that $T\neq \widetilde{T}$, in which case we know by the Buneman theorem that $\Sigma(T) \not= \Sigma(\widetilde{T})$. Using the fact that $\vert \Sigma(T)\vert  = \vert \Sigma(\widetilde{T})\vert  = 2n-3$, we can conclude $\Sigma(T)\setminus \Sigma(\widetilde{T}) \not= \emptyset$. Let $\sigma = A\vert B \in \Sigma(T)\setminus \Sigma(\widetilde{T})$ be of minimal size. Moreover, we may assume without loss of generality that $\vert A\vert  \leq \vert B\vert $, which implies $\vert B\vert \geq \frac{n}{2}$. As all trivial splits are contained in both $\Sigma(T)$ and $\Sigma(\widetilde{T})$, we conclude that $\sigma$ cannot be trivial, and thus we have $\vert A\vert  \geq 2$. Moreover, as $\sigma$ is contained in $\Sigma(T)$, it corresponds to an edge $e=\{u,v\}$ whose removal would divide $T$ into two subtrees $T_A$ and $T_B$, with $T_A$ being further subdividable into $T_{A_1}$ and $T_{A_2}$ by the removal of $u$, cf. Figure \ref{fig:treedecomp} for an illustration.

\begin{figure} 
\center
\includegraphics[width=0.9\textwidth]{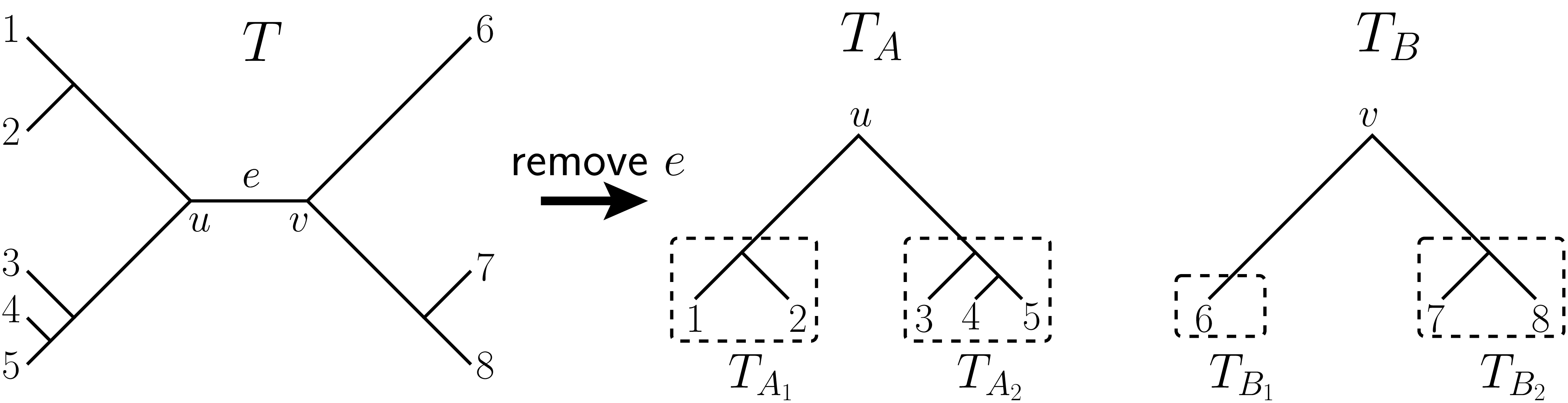}
\caption{  \scriptsize (taken from \cite{Fischer2019, Fischer2022}) By removing an edge $e=\{u,v\}$ from an unrooted binary phylogenetic tree $T$, it is decomposed into two 
rooted subtrees, $T_A$ and $T_B$. If, as in this figure, both of them consist of more than one node, then we can further decompose them into two subtrees $T_{A_1}$ and $T_{A_2}$ or $T_{B_1}$ and $T_{B_2}$, by removing $u$ or $v$, respectively. 
}
\label{fig:treedecomp}
\end{figure}

Note that clearly, $\vert A_1\vert  < \vert A\vert$ and $ \vert A_2\vert  < \vert A\vert$. Furthermore, the two $X$-splits $\sigma_1 = A_1\vert (X\setminus A_1)$ and $\sigma_2 = A_2\vert (X\setminus A_2) $ must both be contained in $\Sigma(T)$ as $T_{A_1}$ and $T_{A_2}$ are subtrees of $T$ inducing these splits. By the minimality of $\sigma$, $\sigma_1$ and $\sigma_2$ are actually  contained in $\Sigma(T)\cap \Sigma(\widetilde{T})$. Therefore, the $X$-splits $\sigma_1$ and $\sigma_2$ are also induced by subtrees of $\widetilde{T}$ (as are all other splits induced by edges in $T_{A_1}$ and $T_{A_2}$, respectively), which implies that $\widetilde{T}$ contains the two  subtrees $T_{A_1}$ and $T_{A_2}$. However, as $\Sigma(\widetilde{T})$ does not contain $\sigma$, there is a path $S = \rho_1, \beta_1, \ldots, \beta_m ,\rho_2$ in $\widetilde{T}$ from the root $\rho_1$ of $T_{A_1}$ to the root $\rho_2$ of $T_{A_2}$ with  $m \geq 2$ (otherwise, $\widetilde{T}$ would also contain $\sigma$), cf. Figure \ref{fig:tildeTfig}.  This also implies that besides $T_{A_1}$ and $T_{A_2}$, there are $m$ more subtrees $\widetilde{T}_i$ obtained by deleting the edges of $S$, with each such subtree $\widetilde{T}_i$ containing $\beta_i$ as a leaf, cf. Figure \ref{fig:tildeTfig}. Moreover, each of the trees $\widetilde{T}_i$ can be thought of as a binary phylogenetic tree with taxon set $X_i = B_i \cup \{\beta_i\}$, where $B_i$ is a subset of $B$.

\begin{figure} 
\center
\includegraphics[width=0.9\textwidth]{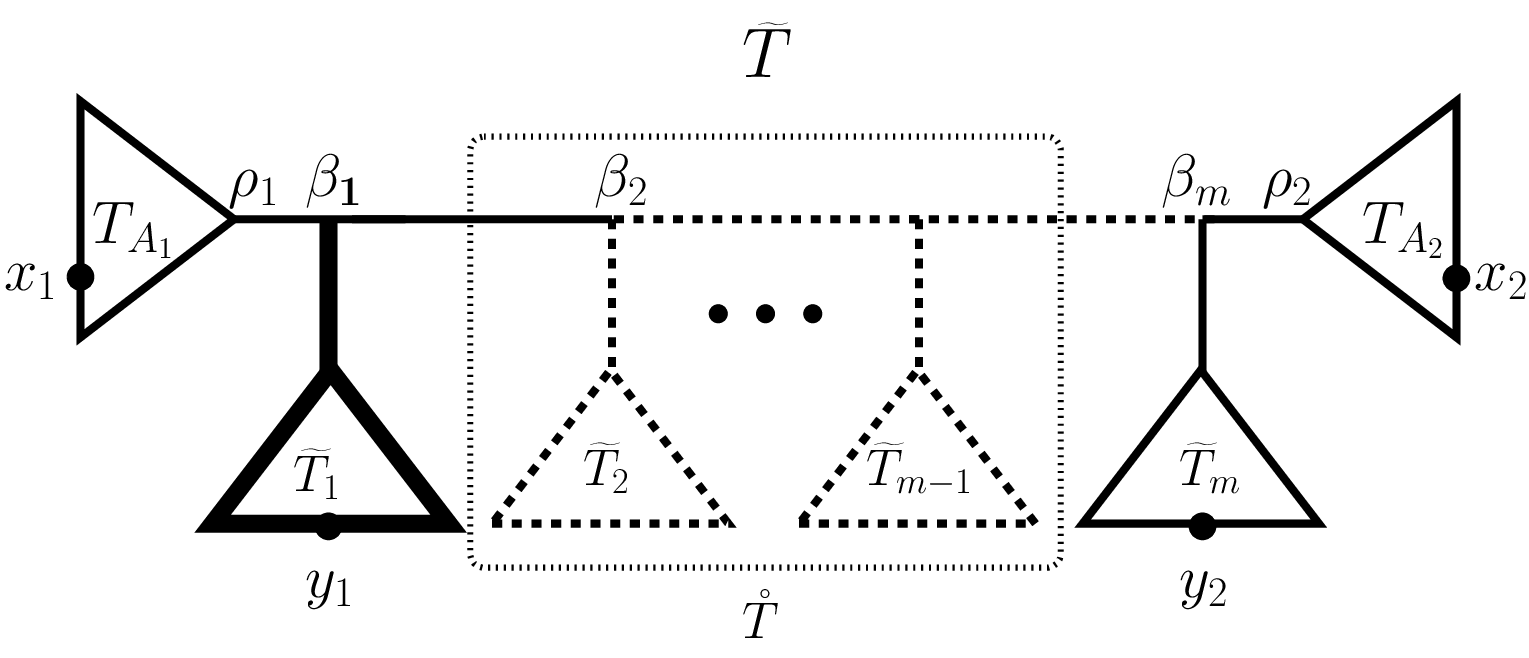}
\caption{  \scriptsize  (adapted from \cite{Fischer2022}) Tree $\widetilde{T}$ as described in the proof of Theorem \ref{thm:characterization}. Subtree $\widetilde{T}_1$ is exemplarily highlighted in bold (note that it does contain $\beta_1$ as a leaf); the other subtrees $\widetilde{T}_i$ are formed analogously. Moreover, note that the subtree $\mathring{T}$ might be empty, i.e., the dashed parts of the tree might not exist, namely if $m=2$.}
\label{fig:tildeTfig}
\end{figure}
Next, our goal is to show that we can construct a binary character $f$ on $X$ which satisfies all of the following properties:
\begin{enumerate}
\item Each taxon in $A$ is assigned state $a$ by $f$. 
\item There are two edge-disjoint paths $P_1$ and $P_2$ with $P_1$ connecting some leaf $x_1\in T_{A_1}$ with some leaf $y_1\in \widetilde{T}_1$ and with $P_2$ connecting some leaf $x_2 \in T_{A_2}$ with some leaf $y_2\in \widetilde{T}_m$, and $f$ assigns both $y_1$ and $y_2$ state $b$.
\item There are additional $k-2$ paths $P_3, \ldots, P_k$ from a leaf in $B$ to another leaf in $B$, respectively, such that $P_1, \ldots, P_k$ form a collection of edge-disjoint paths in $\widetilde{T}$. The endpoints of each path in this collection are assigned different states by $f$.
\item Every taxon in $B$ which is not contained in some path $P_i$ is assigned state $b$ by character $f$.
\end{enumerate}

Note that if we succeed to find a character $f$ with all these properties, then $A_f$ is the union of $A$ and a set of $k-2$ elements of $B$, which are exactly the endpoints of paths $P_3, \ldots, P_k$ to which $a$ is assigned, and $B_f$ is a proper subset of $B$ obtained by deleting the mentioned endpoints from $B$. 

The proof strategy now is as follows: We continue to show that we can indeed choose $f$ as described above. Subsequently, we will show that for this character $f$, we have $f \in A_k(\widetilde{T}) \setminus A_k(T)$. The latter implies $A_k(\widetilde{T}) \neq A_k(T)$ and will thus complete the proof.

So in order to find $f$ fulfilling the above properties, the first and the fourth property do not pose any problem. The crucial point to show is that we can choose the $k$ edge-disjoint paths with the described properties.

First, recalling that each $\widetilde{T}_i$ is a binary  phylogenetic tree with taxon set $X_i=B_i \cup \{\beta_i\}$, we apply Lemma \ref{lem:leafchoice} to $\widetilde{T}_1$ and $\widetilde{T}_m$. To simplify notation, let   $b_i= \vert B_i\vert  + 1$ and $c_i= \left\lfloor\frac{b_i}{2} \right\rfloor$ for $i=1,\ldots,m$. Then, Lemma \ref{lem:leafchoice} implies that $\widetilde{T}_1$ contains at least $c_1$ edge-disjoint leaf-to-leaf paths $R_1, \ldots, R_{c_1}$ such that $\beta_1$ is an endpoint of $R_1$. Analogously, $\widetilde{T}_m$ contains at least $c_m$ edge-disjoint leaf-to-leaf paths $R_1^\prime, \ldots, R_{c_m}^\prime$ with $\beta_m$ being an endpoint of $R_1^\prime$.\\ 

Next, we choose and fix a leaf $x_1$ in $T_{A_1}$ as well as a leaf $x_2$ in $T_{A_2}$. Let $P_1^\prime$ be the unique path from $x_1$ to $\beta_1$ and $P_2^\prime$ be the unique path from $x_2$ to $\beta_m$ in $\widetilde{T}$. Then, let $P_1$ be the concatenation of $P_1^\prime$ and $R_{1}$, i.e., $P_1$ connects $x_1$ in $T_{A_1}$ with some leaf $y_1$ in $\widetilde{T}_1$, and let $P_2$ be the concatenation of $P_2^\prime$ and $R_{1}^\prime$, i.e., $P_2$ connects $x_2$ in $T_{A_2}$ with some leaf $y_2$ in $\widetilde{T}_m$. Clearly, $P_1$ and $P_2$ are edge-disjoint (as they are on opposite sides of edge $\{\beta_1,\beta_2\}$ in $\widetilde{T}$). Thus, by construction of $R_1,\ldots,R_{c_1}$ and $R_1^\prime,\ldots,R_{c_m}$, the set $\{P_1,P_2,R_2,\ldots,R_{c_1},R_2^\prime,\ldots,R_{c_m}^\prime\}$ of $c_1+c_m$ paths is also edge-disjoint. Moreover, clearly $P_1$ and $P_2$ with $x_1$ and $x_2$ assigned state $a$ and $y_1$ and $y_2$ assigned state $b$ fulfill Property 2 of the list above.

 Consider $B^{\prime} := \bigcup\limits_{i=2}^{m-1} B_i$. In the following, let $c^\prime= \left\lfloor \frac{\vert B^\prime \vert}{2}\right\rfloor$. We now distinguish three cases.
 \begin{itemize}
\item If $m=2$, $B^{\prime}$ is empty and thus $\vert B^\prime \vert=0$. In this case, we set $\mathring{T}$ to be the empty tree (with empty leaf set $B^\prime=\emptyset$). In particular, $\mathring{T}$ does not contain any leaf-to-leaf paths, which implies that the number of edge-disjoint leaf-to-leaf paths can be denoted by $0=c^\prime$.
\item If $m=3$, we consider the tree $\widetilde{T}_2^\prime$, which is derived from $\widetilde{T}_2$ by deleting taxon $\beta_2$ and subsequently suppressing the resulting degree-2 vertex. Clearly, the leaf set of $\widetilde{T}_2^\prime$ is precisely $B^{\prime}=B_2$, which is why, by Lemma \ref{lem:Fischer}, it contains at least  $c^\prime$ many edge-disjoint leaf-to-leaf paths. As each of them naturally corresponds to a path in $\widetilde{T}_2$ (by re-introduction of the suppressed degree-2 vertex, if applicable), $\widetilde{T}_2$ contains at least $c^\prime$ many edge-disjoint leaf-to-leaf paths $\mathring{P}_1, \ldots, \mathring{P}_{c^\prime}$.
\item If $m>3$, we define $\mathring{T}$ to be the tree derived from $\widetilde{T}$ by cutting the edges $\{\beta_1,\beta_2\}$ and $\{\beta_{m-1},\beta_m\}$ and keeping only the tree that does neither contain $x_1$ nor $x_2$, cf. Figure \ref{fig:tildeTfig}. Let $\mathring{T}^\prime$ be the tree derived from tree $\mathring{T}$ by suppressing both $\beta_2$ and $\beta_{m-1}$. Then, $\mathring{T}^\prime$ is a binary phylogenetic tree on taxon set $B^\prime$, which is why by Lemma \ref{lem:Fischer}, it contains at least $c^\prime$ many edge-disjoint leaf-to-leaf paths. As each of them naturally corresponds to a path in $\mathring{T}$ (by re-introduction of $\beta_2$ and $\beta_{m-1}$), $\mathring{T}$ contains at least $c^\prime$ many edge-disjoint leaf-to-leaf paths $\mathring{P}_1, \ldots, \mathring{P}_{c^\prime}$.
\end{itemize}
 
 So in all these cases, we can find $c^\prime$ many leaf-to-leaf paths, namely $\mathring{P}_1, \ldots, \mathring{P}_{c^\prime}$,  which by construction are also edge-disjoint to the paths in the set $\{P_1,P_2,R_2,\ldots,R_{c_1},R_2^\prime,\ldots,R_{c_m}^\prime\}$, as none of these paths are using edges $\{\beta_1,\beta_2\}$ or $\{\beta_{m-1},\beta_m\}$, which means that these edges keep the paths apart.

Considering the set $\{R_2,\ldots,R_{c_1},R_2^\prime,\ldots,R_{c_m}^\prime,\mathring{P}_1, \ldots, \mathring{P}_{c^\prime}\}$, it is obvious that if this set of $(c_1-1)+(c_m-1)+c^\prime$ many edge-disjoint paths contains at least $k-2$ paths, then Property 3 of the above list can be fulfilled by a suitable choice of $f$. 

Thus, we next show that $(c_1-1)+(c_m-1)+c^\prime\geq k-2$, or, equivalently, that $c_1+c_m+c^\prime \geq k$. Using $\vert B_1 \vert + \vert B_m \vert+\vert B^\prime \vert=\vert B \vert \geq \frac{n}{2}$ and $k\leq \frac{n}{4}$, we can easily bound the term $c_1 + c_m + c^{\prime}$ as follows:

\begin{align*}
c_1 + c_m + c^{\prime} &= \left\lfloor\frac{b_1}{2} \right\rfloor+ \left\lfloor\frac{b_m}{2}\right\rfloor+ \left\lfloor\frac{\vert B^{\prime}\vert }{2} \right\rfloor \\ &\geq \frac{b_1 - 1}{2} + \frac{b_m - 1}{2} + \frac{\vert B^{\prime}\vert -1}{2}
 \\ &= \frac{\vert B_1\vert  + \vert B_m\vert  + \vert B^{\prime}\vert  - 1}{2}= \frac{\vert B\vert }{2} - \frac{1}{2} \geq \frac{n}{4} - \frac{1}{2} \geq k - \frac{1}{2}
\end{align*}
As $c_1 + c_m + c^{\prime}$ and $k$ are integers, $c_1 + c_m + c^{\prime} \geq k$ follows as desired.

In summary, so far we have shown that we can construct a character $f$ as follows: All taxa in $A$ are assigned state $A$ (Property 1). Taxa $y_1$ and $y_2$, the endpoints of the edge-disjoint paths $P_1$ and $P_2$ in $\widetilde{T}_{1}$ and $\widetilde{T}_{m}$, respectively, are assigned state $b$ (Property 2). There at least $k-2$ more edge-disjoint paths which we refer to as $P_3,\ldots,P_k$ (namely $\{R_2,\ldots,R_{c_1},R_2^\prime,\ldots,R_{c_m}^\prime,\mathring{P}_1, \ldots, \mathring{P}_{c^\prime}\}$), all of which get assigned state $a$ to one endpoint and $b$ to the other endpoint (Property 3). All remaining taxa, if any, get assigned state $b$ (Property 4). So indeed, it is possible to construct a character $f$ with Properties 1--4. 

 It remains to show that the existence of this character $f$ implies $A_k(T)\neq A_k(\widetilde{T})$. In this regard, we will prove that $f \in A_k(\widetilde{T})\setminus A_k(T)$.

Each of the paths $P_1,\ldots, P_k$ connects some taxon in $A_f$ with some taxon in $B_f$, so by Proposition \ref{prop:menger0} we know that $l(f,\widetilde{T}) \geq k$. To show that we also have $l(f, \widetilde{T}) \leq k$ (which then implies that $l(f, \widetilde{T}) = k$), we consider the following extension $g$ of $f$ on $\widetilde{T}$: \begin{equation*} g(v)=\begin{cases} f(v) & \text{if $v \in X$},\\ a & \text{if $v \in \mathring{V}(T_{A_1})\cup \mathring{V}(T_{A_2})$},\\ b  &\text{else.}\end{cases}\end{equation*} Clearly, $g$ induces no changing edges within $T_{A_1}$ or $T_{A_2}$ (as all vertices in these subtrees of $\widetilde{T}$ are assigned state $a$). As the rest of tree $\widetilde{T}$ contains only $k-2$ leaves in state $a$ and all inner vertices outside of $T_{A_1}$ and $T_{A_2}$ are assigned state $b$, $g$ induces at most $k$ changes in total, namely one on edge $\{\rho_1,\beta_1\}$, one on edge $\{\rho_2,\beta_m\}$ and potentially $k-2$ more changes on the pending edges leading to those leaves in state $a$ that are contained in $\bigcup\limits_{i=1}^m \widetilde{T}_i$. The existence of this extension $g$ of $f$ with $ch(g,\widetilde{T}) \leq k$ shows that we indeed have $l(f,\widetilde{T})\leq k$, and thus in total $l(f,\widetilde{T})= k$ and $f \in A_k(\widetilde{T})$ as required.  

 It remains to show that $f \not\in A_k(T)$. In order to see this, like above, we consider an extension $g$ of $f$, but this time on $T$. We define $g$ as follows: \begin{equation*} g(v)=\begin{cases} f(v) & \text{if $v \in X$},\\ a & \text{if $v \in \mathring{V}(T_{A})$},\\ b  &\text{else.}\end{cases}\end{equation*}  Clearly, $g$ induces no changing edges within $T_{A}$ (as all vertices in this subtree of $T$ are assigned state $a$). As the rest of tree $T$ contains only $k-2$ leaves in state $a$ and all inner vertices outside of $T_A$ are assigned state $b$, $g$ induces at most $k-1$ changes in total, namely one on the edge on which $T_A$ is pending (i.e., the edge $\{u,v\}$ connecting the unique vertex $u$ adjacent both to $\rho_1$ and $\rho_2$ in $T$ with its third neighbor $v \neq \rho_1,\rho_2 $)  and potentially $k-2$ more changes on the pending edges leading to those leaves in state $a$ that are contained in $B=X \setminus A$. The existence of this extension $g$ of $f$ with $ch(g,T) \leq k-1$ shows that we indeed have $l(f,T)\leq k-1$, and thus $f \not\in A_k(T)$ as claimed.

In summary, we have found a character $f \in A_k(\widetilde{T})$ for which we know that $f \not\in A_k(T)$, which shows that $A_k(\widetilde{T})\neq A_k(T)$. This completes the proof.
\end{proof}

As stated before, Theorem \ref{thm:characterization} generalizes Proposition \ref{prop:A1A2} to all cases of $k$ with $k \leq \frac{n}{4}$ and thus significantly improves the known bound of $k \leq \frac{n}{20}$ from \cite{Fischer2022}. In the next section, we show that at least within the NNI neighborhood of a binary phylogenetic $X$-tree, the bound can be improved even further in the following sense: If $T$ is a binary phylogenetic tree with $n > 2k+2$ leaves (for $k \in \mathbb{N}_{\geq 1}$) and if $\widetilde{T}$ is an NNI neighbor of $T$, we can guarantee $A_k(T)\neq A_k(\widetilde{T})$.

\subsection{Investigating the NNI neighborhood}

It is the main aim of this section to show that the generalization of Proposition \ref{prop:A1A2} to all cases of $k$ with $k \leq \frac{n}{4}$ provided by the previous section can at least locally be further improved to all $k\leq \frac{n}{2} - \frac{3}{2}$. In particular, we will show that if $k\leq \frac{n}{2}-\frac{3}{2}$ and if $\widetilde{T}$ is an NNI neighbor of $T$, then we have $A_k(T)\neq A_k(\widetilde{T})$.

\begin{theorem}\label{thm:boundNNI} Let $k \in \mathbb{N}_{\geq 1}$ and let $n\in \mathbb{N}_{>2k+2}$. Let $T$ and $\widetilde{T}$ be two binary  phylogenetic $X$-trees with $\vert X\vert=n$ such that $\widetilde{T}$ is in the NNI neighborhood of $T$. Then, $T\cong\widetilde{T}$ if and only if $A_k(T)=A_k(\widetilde{T})$.
\end{theorem}

In order to prove this theorem, we first need to establish some preliminary results. We start with the following definition followed by two technical lemmas, which connect $A_k(T)$ to smaller trees derived by cherry reductions of types 1 and 2, respectively. Such reductions will turn out to be useful for subsequent inductive proofs.

\begin{definition}\label{def:cherryred}
Let $T$ be a binary phylogenetic $X$-tree. Let $T^1$ and $T^2$ be the two trees derived from $T$ when performing a cherry reduction of type 1 or 2, respectively, to cherry $[x,y]$ of $T$. Moreover, let $f: X\to \{a,b\}$ be a binary character.
\begin{enumerate}
\item If $f(x)=f(y)$, we define $f^1$ as the restriction of $f$ to the taxa of $T^1$.
\item If $f(x)\neq f(y)$, we define $f^2$ as the restriction of $f$ to the taxa of $T^2$.
\end{enumerate}
\end{definition}

\begin{lemma}\label{lem:cherryred1} Let $k \in \mathbb{N}_{\geq 1}$, $n \in \mathbb{N}_{\geq 4}$ and let $X=\{1,\ldots,n\}$. Let $T$ be a binary  phylogenetic $X$-tree. Let $T^1$ and $T^2$ be the two trees derived from $T$ when performing a cherry reduction of type 1 or 2, respectively, to cherry $[x,y]$ of $T$. Moreover, let $f \in A_k(T)$. Then, we have:
\begin{enumerate}
\item If $f(x)=f(y)$, then $f^1 \in A_k(T^1)$.
\item If $f(x)\neq f(y)$, then $f^2 \in A_{k-1}(T^2)$.
\end{enumerate}
\end{lemma}

\begin{proof}  
We prove both parts of the lemma separately.

\begin{enumerate}
\item Let $f(x)=f(y)$. As $f \in A_k(T)$, we know $l(f,T)=k$, which by Proposition \ref{prop:menger0} implies that there are $k$ edge-disjoint $A_f$-$B_f$-paths in $T$. This set of edge-disjoint $A_f$-$B_f$-paths naturally corresponds to a set of $k$ edge-disjoint paths in $T^1$ by deleting one leaf of cherry $[x,y]$, say $x$, and suppressing the resulting degree-2 vertex. This procedure might shorten one of the paths by two edges (if it ended in $x$) or by one edge (if it ended in $y$), but as $f(x)=f(y)$, we know that none of the original $A_f$-$B_f$-paths connected $x$ and $y$, so as $n \geq 4$, if such a path ended in $x$ or $y$, it must have had a length of at least 3. So the shortening does not actually delete a path, and if one of the paths in $T$ actually ended in $x$, we can extend it to $y$ in $T^1$ in order to make it a leaf-to-leaf path again. All this shows that in $T^1$, there are $k$ edge-disjoint $A_{f^1}$-$B_{f^1}$-paths. As the deletion of a leaf obviously cannot increase the maximum number of such paths, this implies that the maximum number of edge-disjoint $A_{f^1}$-$B_{f^1}$-paths in $T^1$ is indeed $k$, implying that $l(f^1,T^1)=k$ and thus $f^1\in A_k(T^1)$, which completes the proof of the first assertion.

\item We now consider the case $f(x)\neq f(y)$. As $f \in A_k(T)$, we have $l(f,T)=k$, which by Corollary \ref{cor:edge=vertex} combined with Theorem \ref{thm:menger3} (Menger's Theorem) shows that $A_f$ and $B_f$ cannot be separated by fewer than $k$ vertices. Now, let $P$ be the unique path connecting $x$ and $y$ in $T$. As $f(x)\neq f(y)$, $P$ is an $A_f$-$B_f$-path. We now apply Proposition \ref{prop:menger2} to $\mathcal{P} = \{P\}$ and conclude that there exists a set $\mathcal{Q}$ of $k$ vertex-disjoint $A_f$-$B_f$-paths such that $x$ and $y$ are endpoints of some of the paths. However,  as $\left[x,y\right]$ is a cherry of $T$, this is only possible if $P\in \mathcal{Q}$ (otherwise the two paths ending in $x$ and $y$ would both employ the vertex adjacent to both $x$ and $y$, which would contradict their vertex-disjointness). This immediately implies that when we delete cherry $[x,y]$ by a cherry reduction of type 2 to get $T^2$, $\mathcal{Q}\setminus \{P\}$ is a set of $k-1$ edge-disjoint $A_{f^2}$-$B_{f^2}$-paths in $T^2$. Moreover, this must be the maximum number of edge-disjoint $A_{f^2}$-$B_{f^2}$-paths in $T^2$, because every collection of at least $k$ edge-disjoint paths of $T^2$ combined with $P$ would give a collection of at least $k+1$ edge-disjoint paths in $T$, a contradiction to $k$ being the maximum number of edge-disjoint $A_f$-$B_f$-paths in $T$ by Proposition \ref{prop:menger0}. Thus, we can conclude $l(f^2,T^2)=k-1$ (again by Proposition \ref{prop:menger0}), which shows that $f^2 \in A_{k-1}(T^2)$. This completes the proof of the second assertion.
\end{enumerate}
\end{proof}

Note that the previous lemma shows that some elements of $A_k(T^1)$ and $A_{k-1}(T^2)$ can be derived from the elements of $A_k(T)$. Indeed, the following lemma shows that \textit{all} elements of $A_k(T^1)$ and $A_{k-1}(T^2)$ can be derived from the elements of $A_k(T)$ in this way. Together, Lemma \ref{lem:cherryred1} and Lemma \ref{lem:cherryred2} imply that the opposite is also true, i.e., the elements of $A_k(T)$ can be derived from the elements of $A_k(T^1)$ and $A_{k-1}(T^2)$, respectively.

\begin{lemma}\label{lem:cherryred2} Let $k \in \mathbb{N}_{\geq 1}$, $n \in \mathbb{N}_{\geq 4}$ and let $X=\{1,\ldots,n\}$. Let $T$ be a binary  phylogenetic $X$-tree. Let $T^1$ and $T^2$ the two trees derived from $T$ when performing a cherry reduction of type 1 or 2, respectively, to cherry $[x,y]$ of $T$. Then, we have:
\begin{enumerate}
\item If $g \in A_k(T^1)$, then there exists precisely one character $f \in A_k(T)$ with $f(x)=f(y)$ and $f^1=g$. Moreover, we have $f(x)= f(y)$.
\item If $h \in A_{k-1}(T^2)$, then there exist precisely two characters $f_1$ and $f_2$ in $A_k(T)$ with $f_i(x)\neq f_i(y)$ (for $i=1,2$) and  $f_1^2=f_2^2=h$. Moreover, we have $f_i(x)\neq f_i(y)$ for $i=1,2$.
\end{enumerate}
\end{lemma}

\begin{proof} 
We again prove both assertions separately. 
\begin{enumerate} 
\item Let $g \in A_k(T^1)$ and let $f$ be the character as uniquely defined by Lemma \ref{lem:cherryred2}(1), i.e. $f^1=g$ and $f(x)=f(y)$. Clearly, as $f(x)=f(y)$, no most parsimonious extension of $f$ will ever require a change on cherry $[x,y]$, which shows that $l(f,T)=l(g,T^1)=k$ and thus $f \in A_k(T)$. 

\item Let $h\in A_{k-1}(T^2)$ and let $f_1$ and $f_2$ be as described in Lemma \ref{lem:cherryred2}(2), i.e., $f_1^2=f_2^2=h$ and $f_i(x)\neq f_i(y)$ for $i=1,2$. As the two leaves $x$ and $y$ are in different states in both $f_1$ and $f_2$ and form a cherry in $T$, every most parsimonious extension will require a change on this cherry. This shows that $l(f_1,T)=l(f_2,T)=l(h,T^2)+1=(k-1)+1=k$. Thus, $f_1,f_2 \in A_k(T)$. As clearly there is no other character $\widehat{f}$ on $X$ for which $\widehat{f}$ restricted to $X\setminus \{x,y\}$ equals $h$ and $\widehat{f}^2$ is defined  (i.e., $\widehat{f}(x)\neq \widehat{f}(y)$), this completes the proof.

\end{enumerate}
\end{proof}

The following corollary is a simple conclusion from Lemma \ref{lem:cherryred2}, which will be useful later on.

\begin{corollary}\label{cor:cherryred} Let $k \in \mathbb{N}_{\geq 2}$. Let $T$ be a binary phylogenetic $X$-tree and let $\widetilde{T}$ be an NNI neighbor of $T$. Let $[x,y]$ be a cherry contained in both $T$ and $\widetilde{T}$ such that $A_k(T^1)=A_k(\widetilde{T}^1)$ and $A_{k-1}(T^2)=A_{k-1}(\widetilde{T}^2)$, where $T^1$ and $\widetilde{T}^1$ as well as $T^2$ and $\widetilde{T}^2$ result from $T$ and $\widetilde{T}$ by cherry reductions of types 1 and 2, respectively, using cherry $[x,y]$. Then, we have $A_k(T) = A_k(\widetilde{T})$.
\end{corollary}

\begin{proof} Assume the statement is not true, i.e., assume $A_k(T) \neq A_k(\widetilde{T})$. Let $f \in A_k(T) \setminus A_k(\widetilde{T})$. If we have $f(x)=f(y)$, we perform a cherry reduction of type 1 with $[x,y]$ in $T$ and $\widetilde{T}$ and, by Lemma \ref{lem:cherryred1},  find that $f^1 \in A_k(T^1)=A_k(\widetilde{T}^1)$. On the other hand, $f^1 \in A_k(\widetilde{T}^1)$ implies $f \in A_k(\widetilde{T})$ by Lemma \ref{lem:cherryred2}, a contradiction to the choice of $f$. 

So we must have $f(x)\neq f(y)$. We perform a cherry reduction of type 2 with $[x,y]$ in $T$ and $\widetilde{T}$ and, again by Lemma \ref{lem:cherryred1}, find that $f^2 \in A_{k-1}(T^2)=A_{k-1}(\widetilde{T}^2)$. On the other hand, $f^2 \in A_{k-1}(\widetilde{T}^2)$ implies $f \in A_k(\widetilde{T})$ by Lemma \ref{lem:cherryred2}, again a contradiction to the choice of $f$. 

As both cases lead to a contradiction, we cannot choose such an $f$, which shows $A_k(T)=A_k(\widetilde{T})$.
\end{proof}

Before we can finally prove Theorem \ref{thm:boundNNI}, we need to establish one more preliminary result. In fact, the following proposition turns out to be the main ingredient in our proof. It characterizes all trees  $\widetilde{T}$ in the NNI neighborhood of a binary  phylogenetic tree $T$ for which we have $A_k(T) = A_k(\widetilde{T})$.

\begin{proposition}\label{prop:nni}
Let $k \in \mathbb{N}_{\geq 2}$. Let $T$ be a binary phylogenetic $X$-tree with $n=\vert X\vert \geq 4$ and $A\vert B \in \Sigma^\ast(T)$ inducing subtrees $T_{A_1}$ and $ T_{A_2}$, whose leaves are subsets of $A$, as well as $T_{B_1}$ and $ T_{B_2}$, whose leaves are subsets of $B$, cf. Figure \ref{fig:treedecomp}. Let $n_1 = \vert A_1\vert$, $n_2 = \vert A_2\vert$, $n_3 = \vert B_1\vert$ and $n_4 = \vert B_2\vert$. Moreover, let $\widetilde{T}$ be the tree obtained from $T$ by exchanging $T_{A_2}$ with $T_{B_2}$ (i.e., $T$ and $\widetilde{T}$ are NNI neighbors). Set $s(T,\widetilde{T}) = \sum\limits_{i=1}^4 \left\lfloor \frac{n_i-1}{2}\right\rfloor$. Then, we have: 
\begin{align*} A_k(T) = A_k(\widetilde{T}) & \ \ \Longleftrightarrow \ \  
s(T,\widetilde{T}) < k-2.
\end{align*}
\end{proposition}
Before we continue with the proof of this proposition, we first analyze $s(T,\widetilde{T})$ a bit more in-depth. As by definition, $s(T, \widetilde{T})$ formally depends both on $T$ and the specific $X$-split $A\vert B$, the notation  $s(T, \widetilde{T})$ suggesting that $s$ depends on $T$ and $\widetilde{T}$ may seem counter-intuitive at first. However, if $\widetilde{T}$ is obtained from $T$ by a single NNI move performed on some inner edge $e$ of $T$, then $T$, $\widetilde{T}$ and the $X$-split $A\vert B$ induced by $e$ fulfill the assumptions of Proposition \ref{prop:nni}. So in fact,  $s(T, \widetilde{T})$ depends only on the specific NNI move performed on $T$ to obtain $\widetilde{T}$, which justifies the notation.

\begin{proof} In the following, whenever there is no ambiguity, we refer to $s(T,\widetilde{T})$ simply as $s$. Now, we subdivide the proof into two parts, one for each direction of the statement.

\begin{enumerate}
\item In order to show that $s\geq k-2$ implies $A_k(T)\neq A_k(\widetilde{T})$, we use induction on $k$ and assume $s\geq k-2$.  

Note that for $k=2$, by Proposition \ref{prop:A1A2}, there is nothing to show as in this case, we already know that $T \not\cong \widetilde{T}$ implies $A_k(T) \neq A_k(\widetilde{T})$. This completes the base case of the induction. Therefore, in the following we may assume $k \geq 3$ and that for all pairs $T^\prime_1$, $T^\prime_2$ of binary  phylogenetic $X$-trees that are NNI neighbors with $s(T^\prime_1,T^\prime_2) \geq (k-1)-2=k-3$, we already know that $A_{k-1}(T^\prime_1) \neq A_{k-1}(T^\prime_2)$. Therefore, for any such pair we can assume without loss of generality that there is a character $h \in A_{k-1}(T^\prime_1)\setminus A_{k-1}(T^\prime_2)$ (else we may swap the roles of $T^\prime_1$ and $T^\prime_2$).  
Our aim now is to construct a character $f \in A_k(T)\setminus A_k(\widetilde{T})$, which will, in turn, imply that $A_k(T) \neq A_k(\widetilde{T})$.

So now we have trees $T$ and $\widetilde{T}$ which are NNI neighbors with $s=s(T,\widetilde{T})\geq k-2$. Note that if we had $s=0$, the fact that $s \geq k-2$ would imply $k =2$; i.e., this would refer to the base case of the induction, which we have already considered. Thus, we can now assume $s>0$. By the definition of $s$, this  implies that there is an $i \in \{1,2,3,4\}$ such that $n_i \geq 3$ (because at least one of the summands $\left \lfloor \frac{n_i-1}{2} \right\rfloor$ needs to be at least 1). Without loss of generality, assume $n_1\geq 3$. Then, $T_{A_1}$ contains a cherry $[x,y]$ whose deletion by a cherry reduction of type 2 does not eradicate $T_{A_1}$ (i.e., the remaining tree after the cherry reduction is non-empty). Note that as $T$ and $\widetilde{T}$ differ only in one NNI move (and thus $T_{A_1}$ is subtree in $T$ as well as in $\widetilde{T}$), both trees necessarily contain cherry $[x,y]$. If we denote by $T^\prime$ and $\widetilde{T}^\prime$ the trees resulting from a cherry reduction of type 2 performed on  $[x,y]$, it is clear that  $T^\prime$ and $\widetilde{T}^\prime$ are also NNI neighbors. Moreover, as (compared to $T$ and $\widetilde{T}$) only $n_1$ was reduced by 2 and all other $n_i$ remained unchanged, we have 

\begin{align*} s(T^\prime,\widetilde{T}^\prime) = \left\lfloor \frac{(n_1-2)-1}{2}\right\rfloor+\sum\limits_{i=2}^4 \left\lfloor \frac{n_i-1}{2}\right\rfloor&= \left\lfloor \frac{n_1-1}{2}\right\rfloor-1+\sum\limits_{i=2}^4 \left\lfloor \frac{n_i-1}{2}\right\rfloor\\&=s(T,\widetilde{T})-1\geq k-3,\end{align*}

where the last inequality is due to $s(T,\widetilde{T})\geq k-2$. By the inductive hypothesis, this shows that there exists a character $h \in A_{k-1}(T^\prime)\setminus A_{k-1}(\widetilde{T}^\prime)$. By Lemma \ref{lem:cherryred2}, this implies the existence of two characters $f_1$ and $f_2$ with $f_1^2=f_2^2=h$ and $f_i(x)\neq f_i(y)$ for $i=1,2$ and such that $f_1, f_2 \in A_k(T)$ and $f_1, f_2 \not\in A_k(\widetilde{T})$ (as otherwise we would necessarily have $f_1^2, f_2^2 \in A_{k-1}(\widetilde{T}^{\prime})$).  

This implies that we have found two characters $f_1$ and $f_2$ with $f_1,f_2\in A_k(T)\setminus A_k(\widetilde{T})$, showing that $A_k(T)\neq A_k(\widetilde{T})$ and thus completing the first part of the proof.

\item Next, we need to show that if we have $k \in \mathbb{N}_{\geq 2}$ as well as two trees $T$ and $\widetilde{T}$ which are NNI neighbors with $s(T,\widetilde{T})<k-2$, we have $A_k(T)=A_k(\widetilde{T})$. We will prove this assertion by induction on $s$. In particular, we first show that the statement holds for $s=0$ and all values of $k \geq 3$. (Note that considering $k \geq 3$ indeed covers all values of $k$ that we have to consider, as we have $s\geq 0$, which together with $s<k-2$ implies $k\geq 3$.) Then, we proceed to prove that if the assertion holds for the combinations $s-1$ and $k-1$ as well as $s-1$ and $k$, it also holds for $s$ and $k$.

\begin{itemize}
\item We start with the base case $s=0$. In this case, as $0=s= \sum\limits_{i=1}^4 \left\lfloor \frac{n_i-1}{2}\right\rfloor$, we conclude that $1\leq n_i \leq 2$ for all $i=1,\ldots,4$. We now consider several subcases in order to show that either $A_k(T)=A_k(\widetilde{T})=\emptyset$ or that every element $f \in A_k(T)$ is also contained in $A_k(\widetilde{T})$. By symmetry, the latter will lead to $A_k(T)=A_k(\widetilde{T})$, showing that the two sets are equal in all cases. 

\begin{itemize}
\item If \textit{all} $n_i$ equal 1, we have $n_1+n_2+n_3+n_4=n=4$, i.e., $T$ and $\widetilde{T}$ both have four leaves. As $k \geq 3$ and as there are no binary characters $f$ on four taxa with $l(f,T)\geq 3$ (or $l(f,\widetilde{T})\geq 3$), we have $\emptyset=A_k(T)=A_k(\widetilde{T})$, so indeed both sets are equal. 

\item Next, assume that $n_i=2$ for at least one $i \in \{1,2,3,4\}$ and that $A_k(T) \neq \emptyset$. Let $f \in A_k(T)$. Note that every $i$ for which this is the case implies a taxon set $U_i \in \{A_1, A_2, B_1, B_2\}$ such that $T_{U_i}$ consists of a cherry $[x_i,y_i]$. Now assume all $i$ with $n_i =2$ induce a cherry $[x_i,y_i]$ with $f(x_i)=f(y_i)$. If this was the case, we could apply a cherry reduction of type 1 to $T$ to the first of these cherries to derive a character $f^1 \in A_k(T^1)$ according to Lemma \ref{lem:cherryred1}. Repeating this step for all $i$ with $n_i=2$, we would end up with a character $g \in A_k(T^*)$, where $T^*$ is the four taxon tree resulting from $T$ by iteratively performing cherry reductions of type 1 to the cherries induced by $n_i=2$. However, as $k\geq 3$ and as $T^*$ has only four leaves, as above we know that $A_k(T^*)=\emptyset$, so such a character $g$ cannot exist, which would be a contradiction. 
\item By the above considerations, we know that if we have at least one $i \in \{1,2,3,4\}$ such that $n_i=2$ and $f \in A_k(T)\neq \emptyset$, there must be $U\in \{A_1, A_2, B_1, B_2\}$ with $U = \{x,y\}$ (i.e., $T_U$ consists of the cherry $[x,y]$) such that $f(x) \not= f(y)$. Note that by assumption, $[x,y]$ is a cherry of $\widetilde{T}$, too. We now perform a cherry reduction of type 2 to $[x,y]$, both for $T$ and $\widetilde{T}$. However, note that this eliminates subtree $T_U$ completely (as it only consists of cherry $[x,y]$) and thus leads to $T^2\cong \widetilde{T}^2$. In particular, we have $A_{k-1}(T^2)=A_{k-1}(\widetilde{T}^2)$. However, as $f^2 \in A_{k-1}(T^2)=A_{k-1}(\widetilde{T}^2)$, we know by Lemma \ref{lem:cherryred2} that $f^2$ corresponds to precisely two characters in $A_k(T)$ and $A_k(\widetilde{T})$, namely the two characters that assign different states to $x$ and $y$ and otherwise agree with $f^2$. Clearly, one of these characters is $f$, so we have $f \in A_k(\widetilde{T})$. 
\end{itemize}
So in all possible cases, either $A_k(T) = \emptyset$ or the arbitrarily chosen character $f \in A_k(T)$ is also contained in $A_k(\widetilde{T})$, which shows $A_k(T)\subseteq A_k(\widetilde{T})$. Swapping the roles of $T$ and $\widetilde{T}$ shows the converse inclusion, too, so we conclude $A_k(T)=A_k(\widetilde{T})$ as desired. This completes the proof of the base case $s=0$.

\item We now consider the case $s>0$. We assume the assertion already holds for the pairs $(s-1,k-1)$ as well as $(s-1,k)$ and show that this implies it also holds for the pair $(s,k)$. 

As we are now considering the case $s>0$, we must have $n_i> 2$ for some $i\in \{1,2,3,4\}$. Moreover, it is important to note that with $s >0 $ and  $s<k-2$, we must have $k \geq 4$. 

First, we note that by $n_i> 2$ for some $i\in \{1,2,3,4\}$ the existence of some $T_U$ with $U \in \{A_1,A_2,B_1,B_2\}$ and $\vert U\vert = n_i > 2$ and thus also the existence of some cherry $[x,y]$ contained in $T_U$ (as well as in $T$ and $\widetilde{T}$) is implied.

We first show that for such a cherry $T^1$ and $\widetilde{T}^1$ fulfill the induction hypothesis. Then, using a case distinction, we will analyze in which case also $T^2$ and $\widetilde{T}^2$ fulfill the induction hypothesis and in which case another argument is needed.

As $n_i\geq 3$, a reduction does not completely eliminate $T_U$ (even if we perform a cherry reduction of type 2). So $T^1, T^2, \widetilde{T}^1, \widetilde{T}^2$ can be constructed by applying cherry reductions of types 1 and 2 to cherry $[x,y]$ both in $T$ and $\widetilde{T}$, respectively. Clearly, $\widetilde{T}^1$ is an NNI neighbor of $T^1$ and $\widetilde{T}^2$ is an NNI neighbor of $T^2$. Let $s^1=s(T^1,\widetilde{T}^1)$ and $s^2=s(T^2,\widetilde{T}^2)$. 

Now, we first show that $A_{k-1}(T^2) = A_{k-1}(\widetilde{T}^2)$. Using the fact that $\left \lfloor m-1 \right \rfloor = \left \lfloor m\right \rfloor - 1$ for every $m\in \mathbb{Z}$, we get
\begin{align*}s^2&=s-\left\lfloor \frac{n_i-1}{2} \right\rfloor + \left\lfloor \frac{n_i-3}{2} \right\rfloor\\&= s-\left\lfloor \frac{n_i-1}{2} \right\rfloor + \left(\left\lfloor \frac{n_i-1}{2} \right\rfloor -1\right) =s-1.\end{align*}
By the inductive hypothesis, we thus must have $A_{k-1}(T^2) = A_{k-1}(\widetilde{T}^2)$ (using the assumption on $(s-1,k-1)$).

Finally, we analyze in which case we additionally have $A_k(T^1) = A_k(\widetilde{T}^1)$ by the induction hypothesis and in which case we have to establish this equality by an additional argument.
 In order to do this, we distinguish the following  subcases.

\begin{enumerate}
    \item \label{firstcase} Suppose that additionally to $n_i>2$, we have that $n_i$ is odd. We already know $A_k(T^2) = A_k(\widetilde{T}^2)$. Using $\left\lfloor \frac{n_i-1}{2} \right\rfloor  = \frac{n_i-1}{2}=\left\lfloor \frac{n_i}{2} \right\rfloor$ as $n_i$ is odd, we then have \begin{align*}s^1&=s-\left\lfloor \frac{n_i-1}{2} \right\rfloor + \left\lfloor \frac{n_i-2}{2} \right\rfloor= s- \frac{n_i-1}{2} + \left(\left\lfloor \frac{n_i}{2} \right\rfloor-1 \right)=s-1.\end{align*}

By the inductive hypothesis, we thus must have $A_k(T^1)=A_k(\widetilde{T}^1 )$ (using the assumption on $(s-1,k)$).

Thus, with $A_k(T^1)=A_k(\widetilde{T}^1 )$ and  $A_{k-1}(T^2) = A_{k-1}(\widetilde{T}^2)$, we can use Corollary \ref{cor:cherryred} to conclude $A_k(T)=A_k(\widetilde{T})$.

\item Now,  additionally to $n_i>2$, we assume that $n_i$ is even. As $n_i$ is even, we have $\left\lfloor \frac{n_i-1}{2} \right\rfloor = \left\lfloor \frac{n_i-2}{2} \right\rfloor$, which implies $s^1=s$. So we cannot use the inductive assumption. However, note that if we denote by $n_i^1=n_i-1$ the subtree size of $U$ after the cherry reduction of type 1 in both $T$ and $\widetilde{T}$, then clearly, as $n_i$ is even, we have that $n_i^1$ is odd. This implies that still $n_i^1 > 2$ and we can find a cherry in the reduced subtree fulfilling the assumptions of Case (a). So we can actually apply Case (a) to $T^1$ and $\widetilde{T}^1$ (using $n_i^1$ instead of $n_i$) and conclude that $A_k(T^1)=A_k(\widetilde{T}^1)$. 

So as before, using Corollary 2, we conclude from $A_{k-1}(T^2)=A_{k-1}(\widetilde{T}^2)$ and 
$A_{k}(T^1)=A_{k}(\widetilde{T}^1)$ that $A_k(T)=A_k(\widetilde{T})$.
\end{enumerate}
\end{itemize}
So in all cases, we can conclude that $A_k(T)=A_k(\widetilde{T})$, which completes the second part of the proof.
\end{enumerate}
\end{proof}

Finally, we can now turn out attention to the proof of Theorem \ref{thm:boundNNI}.

\begin{proof}[Proof of Theorem \ref{thm:boundNNI}]
If $A_k(T)\neq A_k(\widetilde{T})$, obviously $T\not\cong \widetilde{T}$, so for the first direction, there is nothing to show. 

So now assume $T\not\cong \widetilde{T}$. We need to show that then, $A_k(T)\neq A_k(\widetilde{T})$. Using the inequality $\left\lfloor \frac{m}{2} \right\rfloor \geq \frac{m-1}{2}$, which holds for all $m \in \mathbb{Z}$, as well as the fact that $n_1+n_2+n_3+n_4=n$, we easily derive the following lower bound for $\sum\limits_{i=1}^4 \left\lfloor\frac{n_i-1}{2}\right\rfloor$:

\begin{align*}
\sum\limits_{i=1}^4 \left\lfloor\frac{n_i-1}{2}\right\rfloor &\geq \sum\limits_{i=1}^4 \frac{n_i-2}{2}=\frac{n_1+n_2+n_3+n_4-8}{2}=\frac{n}{2}-4>k-3,\\
\end{align*}

where the last inequality uses $n> 2k+2$ as assumed by Theorem \ref{thm:boundNNI}. As the left-most sum is an integer, in summary we get $\sum\limits_{i=1}^4 \left\lfloor\frac{n_i-1}{2}\right\rfloor \geq k-2$, which by Proposition \ref{prop:nni} implies $A_k(T)\neq A_k(\widetilde{T})$ and thus completes the proof.
\end{proof}

\subsection{Constructing cases with \texorpdfstring{$n=2k+2$}{n=2k+2} and non-unique \texorpdfstring{$A_k$}{Ak} alignments}

In the previous section, for the case $k\geq 3$ we have seen in Theorem \ref{thm:boundNNI} that if $n>2k+2$  and if $\widetilde{T}$ is in the NNI neighborhood of $T$, we can guarantee that we have $T \cong \widetilde{T}$ if and only if $A_k(T)=A_k(\widetilde{T})$ (note that for the cases $k\in \{1,2\}$, by Proposition \ref{prop:A1A2}, the same equivalence is guaranteed for all $n \in \mathbb{N}_{\geq 1}$, even outside the neighborhood of $T$). Moreover, in \cite{Fischer2019} for every $k>2$ a construction of two trees $T$ and $\widetilde{T}$ with $T \not\cong \widetilde{T}$ and $n=2k$ leaves was shown for which $A_k(T)=A_k(\widetilde{T})$. This leads to the natural question if there exist trees $T$ and $\widetilde{T}$ with $T \not\cong \widetilde{T}$ and with
$A_k(T)=A_k(\widetilde{T})$ if $n=2k+1$ or $n=2k+2$. It is the aim of this section to show that in both of these cases, there indeed exist such trees.

\begin{corollary}\label{cor:NNI1}
For every $k>2$ there exists a pair $T$ and $\widetilde{T}$, $T \not\cong \widetilde{T}$, of binary phylogenetic $X$-trees with $\vert X \vert \in \{2k+1,2k+2\}$ and $A_k(T) = A_k(\widetilde{T})$.
\end{corollary}

\begin{proof} Let $k\geq 3$ and $n=2k+1$ or $n=2k+2$. Let $T_1$ and $T_2$ be the two trees on $n$ leaves depicted in Figure \ref{fig:AkT1NNI}. Then, $T_1$ and $T_2$ are NNI neighbors: A swap of either leaf 4 (if $n=2k+1$) or cherry $[3,4]$ (if $n=2k+2$) with cherry $[5,6]$ around the bold edge $e$ turns $T_1$ into $T_2$. Now let us analyze this case a bit more in-depth. Let $A \vert B $ be the split induced by $e$ in $T$. Let $T_{A_1}$ be the subtree of $T_1$ and $T_2$ containing leaves $1,2$ as well as $9,\ldots n$. Let $T_{A_2}$ denote the subtree of both $T_1$ and $T_2$ that contains taxon 4 (if $n=2k+1$) or the cherry $[3,4]$ (if $n=2k+2$), respectively. Let $T_{B_1}$ denote the subtree of $T_1$ and $T_2$ consisting of cherry $[7,8]$, and let $T_{B_2}$ denote the subtree of $T_1$ and $T_2$ consisting of cherry $[5,6]$. Then, $T_{A_1}$ has $n_1=2+((2k+2)-8)=2k-4$ many leaves, $T_{A_2}$ has $n_2=1$ or $n_2=2$ leaves, respectively, $T_{B_1}$ has $n_3=2$ leaves and $T_{B_2}$  has $n_4=2$ leaves. This leads to

\begin{align*} \sum\limits_{i=1}^4 \left\lfloor \frac{n_i-1}{2}\right\rfloor= \left\lfloor \frac{2k-5}{2}\right\rfloor +\underbrace{\left\lfloor \frac{0 \mbox{ or }1}{2}\right\rfloor}_{=0}+\underbrace{\left\lfloor \frac{1}{2}\right\rfloor}_{=0} +\underbrace{\left\lfloor \frac{1}{2}\right\rfloor}_{=0} =k-3<k-2.\end{align*}

Therefore, by Proposition \ref{prop:nni}, we have $A_k(T_1)=A_k(T_2)$. This completes the proof.
\end{proof}

\begin{figure} 
\center
\includegraphics[width=0.9\textwidth]{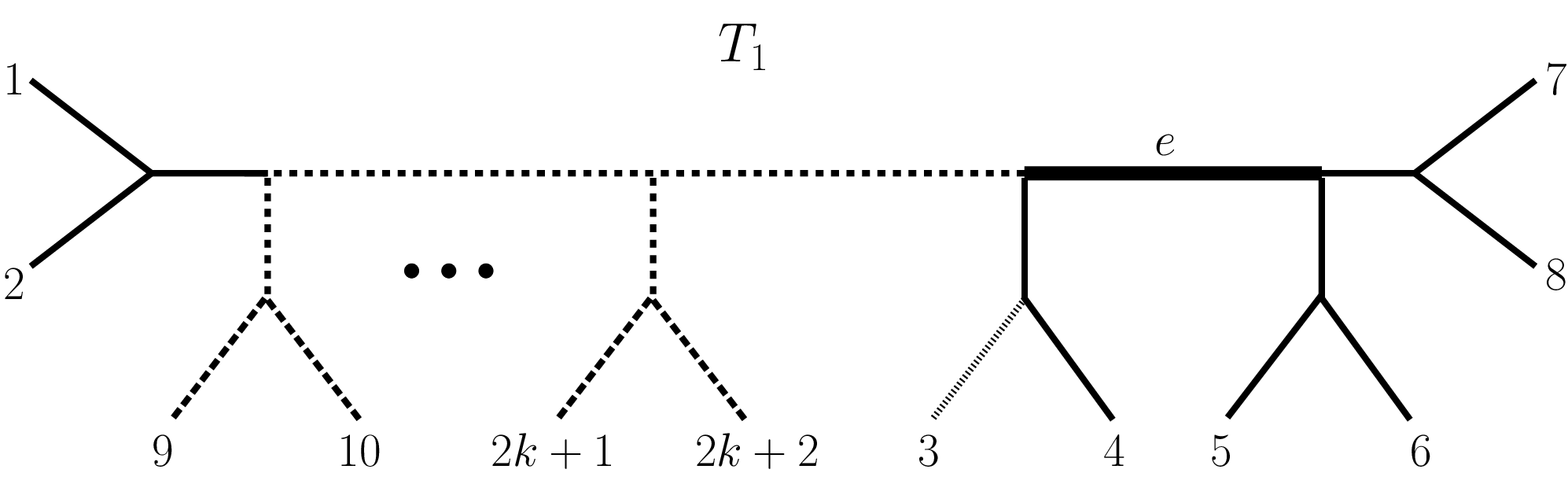}
\par\vspace{.8cm}
\includegraphics[width=0.9\textwidth]{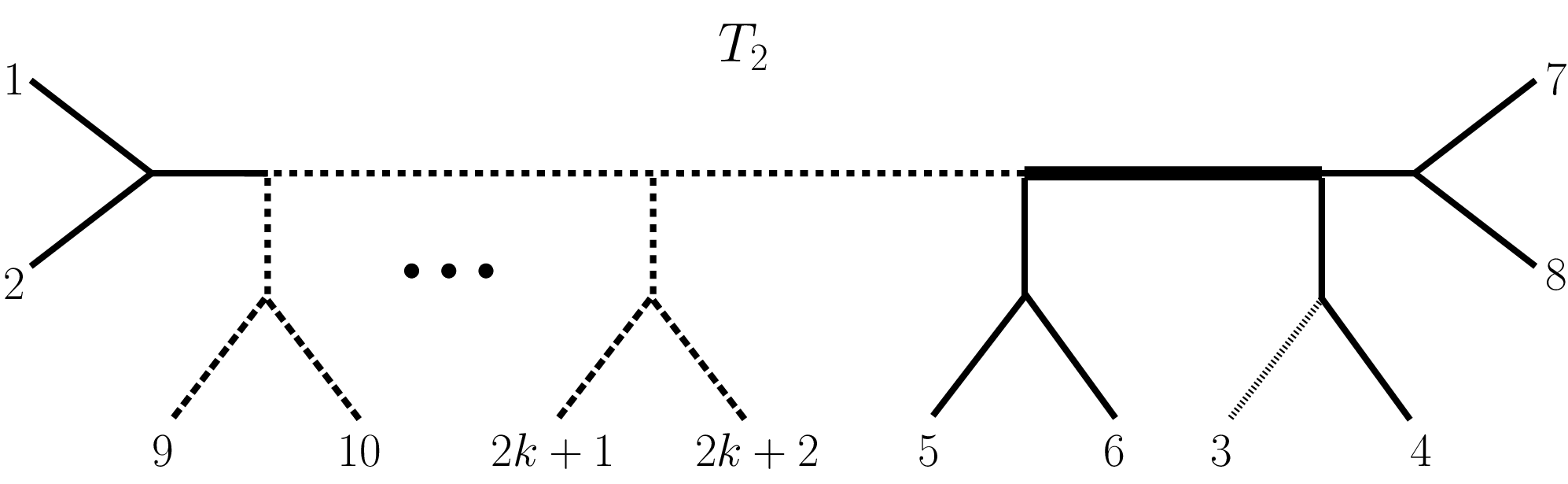}
\caption{  \scriptsize  Two trees $T_1$ and $T_2$. Note that leaf 3 may or may not be there (i.e., there might be $2k+1$ or $2k+2$ leaves in each tree). Note that $T_1$ and $T_2$ are NNI neighbors: A swap of either leaf 4 (if $n=2k+1$) or cherry $[3,4]$ (if $n=2k+2$) with cherry $[5,6]$ around the bold edge $e$ turns $T_1$ into $T_2$. For these trees, we have $A_k(T_1)=A_k(T_2)$ for all $k\geq 3$.}
\label{fig:AkT1NNI}
\end{figure}

\section{Discussion and outlook}
In this manuscript, we have shown that binary phylogenetic trees with $n$ leaves are fully characterized by their $A_k$-alignments whenever $n \geq 4k$ (cf. Theorem \ref{thm:characterization}). Thus, we 
 drastically narrowed the gap resulting from the most recent result in the literature \cite{Fischer2022}, which required $n \geq 20k$, and the fact that was known that for $n=2k$, the statement does not generally hold \cite{Fischer2019}. Additionally, we narrowed the gap further by showing that the statement does not generally hold for up to $n=2k+2$ (cf. Corollary \ref{cor:NNI1}). 
 One very intriguing question for future research is, however, if this gap can be closed completely. We have shown that it can indeed be closed completely for pairs $T$, $ \widetilde{T}$ with $\widetilde{T}$ from the NNI neighborhood of $T$ (cf. Theorem \ref{thm:boundNNI}). In particular, we have shown that the $A_k$-alignment of a binary phylogenetic tree $T$ with $n$ leaves is, for all values of $k\leq \frac{n}{2}-\frac{3}{2}$, unique within the NNI neighborhood of $T$; i.e., no neighbor of $T$ shares the same $A_k$-alignment {with $T$. In fact, we conjecture that the following statement holds, which would close the gap in the interval from $2k+3$ to $4k-1$ affirmatively.

 \begin{conjecture*} Let $k \in \mathbb{N}_{\geq 3}$ and let $T$ and $\widetilde{T}$ binary phylogenetic trees on taxon set $X$ with $\vert X\vert \geq 2k+3$. Then, $T\cong\widetilde{T}$ if and only if $A_k(T) = A_k(\widetilde{T})$. 
 \end{conjecture*}

 Related to this conjecture is the following question: If $\widetilde{T}$ is obtained from $T$ by an NNI operation and given some natural number $k\geq 2$, we call this NNI operation a \textit{problematic move} if it satisfies the condition $s(T,\widetilde{T})<k-2$, where $s(T,\widetilde{T})$ is defined as in Proposition \ref{prop:nni}. Clearly, by Proposition \ref{prop:nni}, if $\widetilde{T}$ is obtained from $T$ by a series of problematic moves, then $A_k(T) = A_k(\widetilde{T})$ (which is why we consider it problematic -- it destroys the uniqueness of the $A_k$-alignment). Now it is tempting to investigate if the converse statement is also true, which leads to the following question.

 \begin{question}\label{qu:2}
     Let $k \in \mathbb{N}_{\geq 3}$ and let $T$ and $\widetilde{T}$ be binary phylogenetic trees on taxon set $X$ and $A_k(T) = A_k(\widetilde{T})$. It is true that $\widetilde{T}$ can be obtained from $T$ by a series of problematic moves?
 \end{question}
 
 By Theorem \ref{thm:boundNNI}, a positive answer to this question would imply the correctness of the above conjecture, because it can easily be seen that for $n\geq 2k+3$ there are no problematic NNI moves.

 Another interesting aspect for future research is tree reconstruction with parsimony based on  $A_k$-alignments. Note that just because an alignment characterizes a tree, this does not necessarily mean that MP (or any other tree reconstruction method) will recover the correct tree. While we know that the unique maximum parsimony tree of $A_1(T)$ is always $T$ (a consequence of the famous Buneman theorem \cite{Buneman1971, Semple2003, Fischer2019}), we know that for $A_2(T)$ this is only true if $n \geq 9$ \cite{Fischer2019}. Exhaustive searches of the tree space performed in  \cite{Goloboff2018} suggest that possibly the latter result can be generalized, i.e., that the unique maximum parsimony tree of $A_k(T)$ might be $T$ whenever $n \geq 4k+1$. This conjecture was formally stated in \cite{Fischer2019}. Note that for MP to be able to recover the \enquote{true} tree from $A_k(T)$, it is a necessary prerequisite that $A_k(T)$ defines $T$. If this was not the case, i.e., if two trees shared the same $A_k$-alignment, there would not be any hope for any tree reconstruction method to distinguish the two trees sharing this alignment from one another.

In summary, as we know that in case $k=2$ we require $n > 4k$ for the reconstruction of trees from their $A_k$-alignments using MP and as we conjecture that this assertion can be generalized to larger values of $k$, the present manuscript, which shows that for $n \geq 4k$ binary phylogenetic trees are defined by $A_k$, is a very useful first step in this regard. Note that the factor of 4 is precisely the one needed to tackle the reconstruction problem. Thus, the improvement from factor 20 to 4 given by the present manuscript is highly relevant, even if the interval $n \in [2k+3,4k-1]$ is still open concerning the characterization of trees by their $A_k$-alignment.

\bibliographystyle{unsrtnat}
\bibliography{references}   

\section*{Declarations}
The authors wish to thank Linda Kn\"uver and Sophie J. Kersting for helpful discussions on the topic. MF also wishes to thank Mike Steel for bringing the topic to her attention.

\section*{Declarations}

\begin{itemize}
\item Funding: No funding was received to assist with the preparation of this manuscript.
\item Conflict of interest/Competing interests: The authors have no competing interests to declare that are relevant to the content of this article.
\item Ethics approval: not applicable 
\item Consent to participate: not applicable 
\item Consent for publication: All authors have given their consent to publish the research findings of the present manuscript.
\item Availability of data and materials: not applicable 
\item Code availability: not applicable
\item Authors' contributions: All authors contributed equally.
\end{itemize}

\end{document}